\documentclass[11pt,lettersize]{article}

\usepackage{amsmath,amssymb}
\usepackage{fullpage}
\usepackage{amsthm}
\usepackage{pdfpages}
\usepackage{mathrsfs}
\usepackage{comment}
\usepackage{times}

\theoremstyle{plain}
\newtheorem{theorem}{Theorem}
\newtheorem{lemma}[theorem]{Lemma}
\newtheorem{proposition}[theorem]{Proposition}
\newtheorem{corollary}[theorem]{Corollary}

\newenvironment{proofof}[1]{ {\noindent \em Proof of #1.}\/}{\hfill\qedsymbol\bigskip}
\newenvironment{proofsketch}{ {\noindent \em Proof sketch.}\/}{\hfill\qedsymbol\bigskip}

\newcommand{\remove}[1]{}
\newcommand{\suppress}[1]{}

\newcommand{\NN}{{\mathbb{N}}}
\newcommand{\RR}{{\mathbb{R}}}

\newcommand{\eps}{\epsilon}

\DeclareMathOperator{\cost}{cost}
\DeclareMathOperator{\floor}{floor}
\DeclareMathOperator{\card}{card}
\DeclareMathOperator{\cm}{cm}
\DeclareMathOperator{\ctr}{ctr}
\DeclareMathOperator{\activ}{active}
\DeclareMathOperator{\pclusters}{preclusters}
\DeclareMathOperator{\mclusters}{metaclusters}
\DeclareMathOperator{\clusters}{clusters}
\DeclareMathOperator{\opt}{opt}

\DeclareMathOperator{\last}{last}

\DeclareMathOperator{\poly}{poly}
\DeclareMathOperator{\dist}{dist}

\usepackage[utf8]{inputenc}
\usepackage{microtype}
\usepackage{algorithm}
\usepackage{algpseudocode}

\begin{document}

\title{Min-Sum Clustering (with Outliers)}

\author{Sandip Banerjee\thanks{The Hebrew University of Jerusalem, {\tt sandip.ndp@gmail.com}. Research supported in part by Yuval Rabani's NSFC-ISF grant 2553-17.} \and
Rafail Ostrovsky\thanks{University of California, Los Angeles, {\tt rafail@cs.ucla.edu}. Supported in part by DARPA under Cooperative Agreement No: HR0011-20-2-0025, NSF Grant  CNS-2001096, US-Israel BSF grant  2015782, Google Faculty Award, JP Morgan Faculty Award, IBM Faculty Research Award, Xerox Faculty Research Award, OKAWA Foundation Research Award, B. John Garrick Foundation Award, Teradata Research Award, and Lockheed-Martin Corporation Research Award. The views and conclusions contained herein are those of the authors and should not be interpreted as necessarily representing the official policies, either expressed or implied, of DARPA, the Department of Defense, or the U.S. Government. The U.S. Government is authorized to reproduce and distribute reprints for governmental purposes not withstanding any copyright annotation therein.} \and
Yuval Rabani\thanks{The Hebrew University of Jerusalem, {\tt yrabani@cs.huji.ac.il}. Research supported in part by NSFC-ISF grant 2553-17 and by NSF-BSF grant 2018687.}
}

\maketitle

\begin{abstract}
We give a constant factor polynomial time pseudo-approximation algorithm for min-sum 
clustering with or without outliers. The algorithm is allowed to exclude an arbitrarily small
constant fraction of the points. For instance, we show how to compute a solution that 
clusters 98\%  of the input data points and pays no more than a constant factor times the 
optimal solution that clusters 99\% of the input data points. More generally, we give the
following bicriteria approximation: For any $\eps > 0$, for any instance with $n$ input points
and for any positive integer $n'\le n$, we compute in polynomial time a clustering of at 
least $(1-\eps) n'$ points of cost at most a constant factor greater than the optimal cost 
of clustering $n'$ points. The approximation guarantee grows with $\frac{1}{\eps}$. Our 
results apply to instances of points in real space endowed with squared Euclidean distance, 
as well as to points in a metric space, where the number of clusters, and also the dimension 
if relevant, is arbitrary (part of the input, not an absolute constant).
\end{abstract}

\thispagestyle{empty}
\newpage
\setcounter{page}{1}

\section{Introduction}

We consider min-sum $k$-clustering. This is the problem of partitioning an input 
dataset of $n$ points into $k$ clusters with the objective of minimizing the sum 
of intra-cluster pairwise distances. We consider primarily the prevalent setting of 
instances of points in $\RR^d$ endowed with a distance function equal to the
squared Euclidean distance (henceforth refered to as the $\ell_2^2$ case). Our 
results apply also to the case of instances of points endowed with an explicit 
metric (henceforth refered to as the metric case). Note that we consider $k$ 
(and $d$, if relevant) to be part of the input, rather than an absolute constant. 
In these and
similar cases we give polynomial time approximation algorithms that cluster all
but a negligible constant fraction of outliers at a cost that is at most a constant
factor larger than the optimum clustering. More specifically, for any $\eps > 0$,
if the optimum we compete against is required to cluster any number $n'\le n$
of points, our algorithm clusters at least $(1-\eps) n'$ points and at most $n'$
points, and pays a constant factor more than the optimum for $n'$ points. The
constant depends on $\eps$.

Clustering in general is a fundamental question in unsupervised learning. The
question originated in the social sciences and now is widely applicable in data
analysis and machine learning, in areas including bioinformatics, computer
vision, pattern recognition, signal processing, fraud/spam/fake news filtering, 
and market/population segmentation. Clustering is also a list of fundamental 
discrete optimization problems in computational geometry that have been 
studied for decades by theoreticians, in particular (but not exclusively) as 
simple non-convex targets of machine learning. Some clustering problems, 
notably centroid-based criteria such as $k$-means, have been studied 
extensively. We currently have a fairly tight analysis of their complexity in the 
worst case (e.g.~\cite{ADHP09,MNV09,ANSW17,CK19}) and under a wide 
range of restrictive conditions: low dimension (e.g., ~\cite{HK07,CKM16,FRS16})
fixed $k$ (e.g.~\cite{KSS05,FMS07,Che09,FL11}), various notions of stability 
(e.g.~\cite{ORSS06,ABS10,KK10,AMR11,CS17}), restrictive models of
computation (e.g.,~\cite{BO07,AJM09,BMORST11,BFL16}), etc., as well as 
practically appealing heuristics (e.g., Lloyd's 
iteration, local search) and supportive theoretical justification 
(e.g., some of the afore-mentioned papers and also~\cite{AV07,JG12}).

Theoretical understanding of density-based clustering criteria, and in particular
min-sum clustering, is far less developed. There are clearly cases in practice
where, for instance, min-sum clustering coincides far better with the intuitive
clustering objective than standard centroid-based criteria. A simple illustrative
example is the case of separating two concentric dense rings of points in the 
plane. Moreover, min-sum clustering satisfies Kleinberg's consistency axiom,
whereas a fairly large class of centroid-based criteria including $k$-means
and $k$-median do not satisfy this axiom~\cite{Kle03,ZB09}.

However, the state-of-the-art for computing min-sum clustering remains
inferior to alternatives. Min-sum $k$-clustering is NP-hard in the $\ell_2^2$
case (e.g., using arguments from~\cite{ADHP09}), and also for the metric
case (see~\cite{FK01}), even for $k=2$. In the $\ell_2^2$ case, it can be 
solved in polynomial time if both $d$ and $k$ are absolute constants~\cite{IKI94}. 
In the metric case with arbitrary $k$, approximating min-sum clustering to within 
a factor better than $1.415$ is NP-hard~\cite{GI03,CKL19}. If $k$ is a fixed constant, 
the problem admits a PTAS, both in the $\ell_2^2$ case and in the metric case~\cite{FKKR03}; 
see also~\cite{GH98,Ind99,Mat00,Sch00} for previous work in this vein. If $k=o(\log n/\log\log n)$, 
then there is a constant factor approximation algorithm for the $\ell_2^2$
case~\cite{CS07}. In the metric case, assuming that $k=o(\log n/\log\log n)$
and the instance satisfies a certain clusterability/stability condition, a partition
close to optimal can be computed in polynomial time~\cite{BBG09,BB09} (see
also~\cite{VBRTX11} for some applications and experimental results in this vein).
We note that practical applications often require $k$ in the thousands, so the
above restrictions on $k$ are unrealistic in those cases.

In the worst case, and under no restrictions on the instance, the best known 
approximation guarantee known is an $O(\log n)$ approximation 
algorithm~\cite{BFSS15} for the metric case. This improves upon a slightly 
worse and much earlier guarantee~\cite{BCR01}. In both papers, the factor 
is derived from representing the input metric space approximately as a convex 
combination of hierarchically separeted tree (HST) metrics~\cite{Bar96,Bar98,FRT03}. 
This incurs logarithmic distortion, which is asymptotically tight in the worst 
case. In HST metrics, min-sum clustering can be approximated to within a 
constant factor. Thus, a fundamental challenge of the study of min-sum
clustering is to eliminate the gap between the hardness of approximation lower
bound of $1.415$ and the approxiomation guarantee upper bound of $O(\log n)$.
We show that a constant factor approximation is possible, if one is willing to
regard as outliers a small fraction of the input dataset. For the $\ell_2^2$ case,
we are not aware of any previous non-trivial guarantee for $k\gg \log n/\log\log n$.

Our results are derived using a reduction from min-sum clustering to a
centroid-based criterion with (soft) capacity constraints. This can be done
exactly in the $\ell_2^2$ case, and approximately in the metric case, though
to get polynomial time algorithms we use an approximation in both cases.
This reduction underlies also some of the above-mentioned previous results
on min-sum clustering.
The outcome of this reduction is a $k$-median or $k$-means problem with
non-uniform capacities. If we are aiming for a constant factor approximation 
then we can afford to violate the capacities by a constant factor. There are 
nice results on approximating $k$-median with non-uniform capacities, for
instance~\cite{DL16}. Unfortunately, these results do not seem applicable 
here, because their input is a metric space. The reduction, even for the 
metric case, does not generate a metric instance of capacitated $k$-median
(the triangle inequality is violated unboundedly).
Nevertheless, we do draw some ideas from this literature. 

Our min-sum
clustering algorithm is based on the the well-trodden path of using the
primal-dual schema repeatedly to search for a good Lagrange multiplier
in lagrangian relaxation of the problem (see~\cite{JV01} for the origin of
this method). The dual program has a variable for every data point, and 
a constraint for every possible cluster. The dual ascent process requires 
detection of constraints that become tight. In our case, this is a non-trivial 
problem, which we solve only approximately. As usual, the dual values are 
used to ``buy'' the opening of the clusters that become tight, and we have 
to contend with points contributing simultaneously to multiple clusters. 
This is done, as usual, by creating a conflict graph among the tight clusters
and choosing an independent set in this graph. However, in our case
there are unsual complications. The connection cost is a distance (not
a metric in the $\ell_2^2$ case, but this is a minor concern) multiplied
by the cardinality of the cluster. If there is a conflict between a large
cluster and some small clusters, we have the following dilemma. If we
open the large cluster, the unclustered points in the small clusters may
lack dual ``money'' to connect to the large cluster; they can only afford
the distance multiplied by the cardinality of their small cluster. If, on the
other hand, we open (some or all of) the small clusters, assigning the
unclustered points in the large cluster to those small clusters might
inflate their cardinality by a super-constant factor, leaving all points with
insufficient funds to connect to the inflated clusters.

We resolve this dilemma as follows (using in part some ideas from~\cite{CR05}).
We open larger clusters first, so if a cluster is not opened, it is smaller than
the conflicting cluster that blocked it. Unclustered points are not assigned to
the blocking cluster, but rather aggregated around each blocking cluster to
form their own clusters of appropriate cardinality. We use approximate
cardinality, in scales which are powers of a constant $b$. As we require
the Lagrange multiplier preserving (LMP) property, we must have sufficient 
``funds'' to pay the opening costs in full (but can setttle for paying just a fraction 
of the connection cost). This is possible if in a scale of, say, $b^p$ we have,
say, at least $b^{2+p}$ unclustered points in clusters of scaled cardinality
$b^p$ (each set of roughly this size can afford to open its own cluster). If a 
blocking cluster is blocking fewer points in this scale, we can't afford to cluster 
them and must discard them as outliers. This is the primary source of the excess 
outliers.

As usual, the search for a good Lagrange multiplier may end with two integer 
primal solutions, one with fewer than $k$ clusters and one with more than $k$ 
clusters, whose convex combination is a feasible fractional bipoint solution to 
the $k$-clustering problem. In our case, as we already may have to give up on 
some outliers, we can simply output either the $<k$ solution or the $k$ largest 
clusters in the $>k$ solution. We point out that these extra outliers can be
avoided by using a more sophisticated ``rounding'' of the bipoint solution, but
given our loss in the primal-dual phase, it would not improve meaningfully our 
guarantees.

The above description sums up the algorithm in the case that $n' = n$.
Our result also extends to the case that the optimal solution is also allowed
to discard some outliers (we may have to discard some more). The main
additional issue in the case $n' < n$ is that in the primal-dual phase we
may open a cluster that brings the number of clustered points from below
$n'$ to above $n'$. In this case, some points in this last cluster need to be 
discarded, but then the remaining clustered points might have insufficient 
``funds'' to open the last cluster. If we have many clusters, we can afford to
eliminate the smallest cluster, declaring its points as outliers, and use the dual 
values of the points in that cluster to pay for opening the last one. If there is a 
small number of clusters, we may assume that the primal-dual phase opened 
less than $k$ clusters (to ensure this property, if $k$ is a small constant, we 
employ the known PTAS; thus we can assume that $k$ is large). Our approach 
in this case draws from~\cite{AS16}, where a similar issue is addressed in the 
case of the sum-of-radii $k$-clustering problem. Though the questions are quite 
different, we use a similar idea of computing a (slightly) non-Lagrange multiplier 
preserving approximation to the lagrangian relaxation. The LMP property is regularly 
used in the argument that the bipoint solution is both feasible and cheap; the approach 
we adopt requires an extra argument to bound the cost of a bipoint solution that 
incorporates a non-LMP solution.

The rest of the paper is organized as follows. Section~\ref{sec: def} introduces
some basic definitions and claims. Section~\ref{sec: alg} describes the algorithm.
Section~\ref{sec: proof} analyzes the algorithm. For conciseness, the paper
presents the $\ell_2^2$ case. Our main result is Theorem~\ref{thm: main sqEuc}.
The metric case is essentially identical, and is
briefly explained in Theorem~\ref{thm: main metric}. We note that we made 
no effort to optimize the constant factor guarantees, throughout the paper.

\section{Definitions and Preliminary Claims}\label{sec: def}

Consider an instance of min-sum clustering that is defined by a set of points
$X\subset\RR^d$ and a target number of clusters $k\in\NN$. Let $n = |X|$.
The cost of a cluster $Y\subset X$ is
$$
\cost(Y) = \frac 1 2\cdot \sum_{x,y\in Y} \|x - y\|_2^2.
$$
The center of mass (or mean) of $Y$ is $\cm(Y) = \frac{1}{|Y|}\cdot\sum_{x\in Y} x$.
The following proposition is a well-known fact (for instance, see~\cite{IKI94}).
\begin{proposition}\label{pr: reduction to cap}
The following assertions hold for every finite set $Y\subset\RR^d$.
\begin{enumerate}
\item The center of mass $\cm(Y)$ is the unique minimizer of $\sum_{x\in Y} \|x - y\|_2^2$ over $y\in\RR^d$.
\item $\cost(Y) = |Y|\cdot\sum_{x\in Y} \|x - \cm(Y)\|_2^2$.
\end{enumerate}
\end{proposition}

A {\em min-sum $k$-clustering} of $X$ is a partition of $X$ into $k$ disjoint 
subsets $X_1,X_2,\dots,X_k$ that minimizes over all possible partititions
$$
\sum_{i=1}^k \cost(X_i) = \sum_{i=1}^k |X_i|\cdot\sum_{x\in X_i} \|x - \cm(X_i)\|_2^2.
$$
In the version allowing outliers, we are given a target $n'\le n$ of the number
of points to cluster, and we are required that $\left|\bigcup_{i=1}^k X_i\right|\ge n'$.
Clearly, the version without outliers is a special case of the version with outliers
with $n' = n$. Let $\opt(X,n',k)$ denote the optimal min-sum cost of clustering
$n'$ points in $X$ into $k$ clusters.
Formally, we can express the goal as a problem of optimizing an exponential
size integer program:
\begin{equation}\label{eq: IP}
\begin{array}{lll}
\hbox{minimize} & \sum_{Y\subset X} \cost(Y)\cdot z_Y & \\
\hbox{s.t.} & \sum_{Y\ni x} z_Y + w_x\ge 1 & \forall x\in X \\
 & \sum_{Y\subset X} z_Y \le k & \\
 & \sum_{x\in X} w_x \le n - n' & \\
 & z_Y\in\{0,1\} & \forall Y\subset X \\
 & w_x\in\{0,1\} & \forall x\in X.
\end{array}
\end{equation}
Fix $b\in\NN$, $b > 1$. For $i\in\NN$, let $\floor_b(i) = b^{\lfloor\log_b i\rfloor}$.
For $Y\subset X$, let $\ctr(Y)$ be a reference point that we set for now as
$\ctr(Y) = \cm(Y)$. Define 
$$
\cost_b(Y) = \floor_b(|Y|)\cdot\sum_{y\in Y} \|y - \ctr(Y)\|_2^2.
$$
In other words, (assuming $\ctr(Y) = \cm(Y)$) we revise $\cost(Y)$ by rounding 
$|Y|$ down to the nearest power of $b$. Thus, 
$\frac{1}{b}\cdot \cost(Y) < \cost_b(Y) \le\cost(Y)$.
We relax the integer program~\eqref{eq: IP} as follows ($b$ to be determined
later):
\begin{equation}\label{eq: LP}
\begin{array}{lll}
\hbox{minimize} & \sum_{Y\subset X} \cost_b(Y)\cdot z_Y & \\
\hbox{s.t.} & \sum_{Y\ni x} z_Y + w_x\ge 1 & \forall x\in X \\
 & \sum_{Y\subset X} z_Y \le k & \\
 & \sum_{x\in X} w_x \le n - n' & \\
 & z_Y\ge 0 & \forall Y\subset X \\
 & w_x\ge 0 & \forall x\in X.
\end{array}
\end{equation}
Then, following a well-traveled path, we lagrangify the constraint on the
number of clusters to get the following lagrangian relaxation ($\lambda$ 
denotes the unknown Lagrange multiplier).
\begin{equation}\label{eq: LP-lag}
\begin{array}{lll}
\hbox{minimize} & \sum_{Y\subset X} \cost_b(Y)\cdot z_Y + 
                  \lambda\cdot \left(\sum_{Y\subset X} z_Y - k\right) & \\
\hbox{s.t.} & \sum_{Y\ni x} z_Y + w_x\ge 1 & \forall x\in X \\
 & \sum_{x\in X} w_x \le n - n' & \\
 & z_Y\ge 0 & \forall Y\subset X \\
 & w_x\ge 0 & \forall x\in X.
\end{array}
\end{equation}
For fixed $\lambda$, this is a linear program, and its dual is:
\begin{equation}\label{eq: dual-lag}
\begin{array}{lll}
\hbox{maximize} & \sum_{x\in X} \alpha_x - \gamma\cdot (n-n') - \lambda\cdot k & \\
\hbox{s.t.} & \sum_{x\in Y} \alpha_x\le \lambda + \cost_b(Y) & \forall Y\subset X \\
 & 0\le \alpha_x \le \gamma & \forall x\in X.
\end{array}
\end{equation}
Notice that the linear program~\eqref{eq: LP-lag} can be interpreted as a relaxation
of the ``facility location'' version of the problem, with $\lambda$-uniform cluster 
opening costs.
\begin{lemma}\label{lm: duality}
For any $\lambda$, the optimal value of the linear program~\eqref{eq: dual-lag} is 
a lower bound on the optimal value of the integer program~\eqref{eq: IP}.
\end{lemma}

\begin{proof}
Consider any optimal solution $(z,w)$ to the integer program~\eqref{eq: IP}.
Notice that we may assume that $\sum_{Y\subset X} z_Y = k$, otherwise we
can split some clusters to get exactly $k$ of them. Splitting clusters cannot
increase the cost of the solution. This is also a feasible solution to the linear 
program~\eqref{eq: LP-lag}. Moreover, the Lagrange term
$\lambda\cdot \left(\sum_{Y\subset X} z_Y - k\right)$ zeroes out, and
$\sum_{Y\subset X} \cost_b(Y)\cdot z_Y\le \sum_{Y\subset X} \cost(Y)\cdot z_Y$.
By weak duality, the value of any feasible solution to the dual program~\eqref{eq: dual-lag}
is a lower bound on the value of any feasible solution to the linear program~\eqref{eq: LP-lag}.
\end{proof}

An obvious issue with the dual program~\eqref{eq: dual-lag} is that the number of
constraints is exponential in $n$. We want to construct a dual solution by growing 
the dual variables, however, it is not clear how to detect new tight dual constraints
without enumerating over the $\exp(n)$ number of constraints. We now address 
this issue. First consider the following fact.
\begin{proposition}\label{pr: eps net}
Let $Y$ be a finite set of points in $\RR^d$. There exists $y\in Y$ such that
$\sum_{x\in Y} \|x - y\|_2^2 \le 2\cdot\sum_{x\in Y}\|x - \cm(Y)\|_2^2$.
(We note that the factor of $2$ can be improved to $1+\eps$, for any $\eps > 0$,
using the center of mass of $O(1/\eps^2)$ points in $Y$, e.g.~\cite{IKI94,FKKR03}.)
\end{proposition}

\begin{proofsketch}
Notice that for every $y\in\RR^d$,
$$
\sum_{x\in Y} \|x - y\|_2^2 \le \sum_{x\in Y}\|x - \cm(Y)\|_2^2 + |Y|\cdot \|y - \cm(Y)\|_2^2
$$
(see, e.g.~\cite{ORSS06}).
Thus, by picking $y\in Y$ that minimizes $\|y - \cm(Y)\|_2^2$, the proposition follows.
\end{proofsketch}

An immediate consequence of Proposition~\ref{pr: eps net} is that ${\cal F} = X$ is 
a set of $n$ points in $\RR^d$, such that for every $Y\subset X$ there exists a point 
$c_Y\in {\cal F}$ such that 
\begin{equation}\label{eq: approx center}
\sum_{x\in Y} \|x - \cm(Y)\|_2^2\le \sum_{x\in Y} \|x - c_Y\|_2^2\le 2\cdot\sum_{x\in Y} \|x - \cm(Y)\|_2^2. 
\end{equation}
(We can improve the factor of $2$ to any constant $1+\eps$ by increasing the size of 
${\cal F}$ to $n^{O(1/\eps^2)}$.) Now, given ${\cal F}$, set initially $\ctr(Y) = c_Y$ for all 
$Y\subset X$. Notice that this puts 
$\cost_b(Y) = \floor_b(|Y|)\cdot\sum_{y\in Y} \|y - c_Y\|_2^2$. We consider the following 
revised dual program.
\begin{equation}\label{eq: dual-lag-rev}
\begin{array}{lll}
\hbox{maximize} & \sum_{x\in X} \alpha_x - \gamma\cdot (n-n') - \lambda\cdot k & \\
\hbox{s.t.} & \sum_{x\in Y} \alpha_x\le \lambda + 
        \floor_b(|Y|)\cdot\sum_{x\in Y} \|x - y\|_2^2 & \forall Y\subset X,\ \forall y\in Y \\
   & 0\le \alpha_x \le \gamma & \forall x\in X.
\end{array}
\end{equation}
\begin{lemma}\label{lm: duality-rev}
For any $\lambda$, the optimal value of the linear program~\eqref{eq: dual-lag-rev} is 
at most twice the optimal value of the integer program~\eqref{eq: IP}.
\end{lemma}

\begin{proof}
The dual of the linear program~\eqref{eq: dual-lag-rev} is
\begin{equation}\label{eq: LP-lag-rev}
\begin{array}{ll}
\min & \left\{\sum_{Y\subset X}\sum_{y\in Y} \floor_b(|Y|)\cdot\sum_{x\in Y} \|x - y\|_2^2\cdot z_{Y,y} + 
                  \lambda\cdot \left(\sum_{Y\subset X}\sum_{y\in Y} z_{Y,y} - k\right):\right.  \\
  & \left.\forall x\in X,\ \sum_{Y\ni x}\sum_{y\in Y} z_{Y,y} + w_x\ge 1\wedge
 \sum_{x\in X} w_x \le n - n' \wedge z,w\ge 0\right\}
\end{array}
\end{equation}
Consider an optimal clustering of any $n'$ points in $X$ into $k$ disjoint clusters 
$Y_1,Y_2,\dots,Y_k$. For all $Y\subset X$, set $z_{Y,y}$ to be the indicator
that $Y$ is a cluster in this list and $y = c_Y$. Also, for all $x\in X$ set $w_x$
to be the indicator that $x$ is not clustered.  Clearly, this is a feasible solution
to the linear program~\eqref{eq: LP-lag-rev}, so its value is an upper bound on
the optimal value of the linear program~\eqref{eq: dual-lag-rev}. The Lagrange
term vanishes as there are exactly $k$ non-zero values $Z_{Y,y}$. Thus, the
upper bound is
$$
\sum_{j=1}^k \floor_b(|Y_j|)\cdot\sum_{x\in Y_j} \|x - c_{Y_j}\|_2^2\le
2\cdot\sum_{j=1}^k \floor_b(|Y_j|)\cdot\sum_{x\in Y_j} \|x - \cm(Y_j)\|_2^2\le
2\cdot\sum_{j=1}^k \cost(Y_j),
$$
where the first inequality uses Equation~\eqref{eq: approx center}.
\end{proof}

In the primal-dual procedure, there is an active set $\activ\subset X$ of points for
which it is safe to raise the dual variable $\alpha_x$ for all $x\in\activ$. We need
to detect when a new dual constraint becomes tight and requires the removal of
the points that are involved from $\activ$. This can be done in polynomial time
for the revised dual program~\eqref{eq: dual-lag-rev} as follows. For every $y\in X$
and for every $j\in\{0,1,2,\dots,\log_b\floor_b(n)\}$, we check if there exists
$Y\subset X$ that satisfies ($i$) $y\in Y$; ($ii$) $Y\cap \activ\ne\emptyset$;
($iii$) $\log_b \floor_b(|Y|) = j$; ($iv$) $\sum_{x\in Y} \alpha_x \ge \lambda + 
b^j\cdot\sum_{x\in Y} \|x - y\|_2^2$. 
In order to do this, consider the set of points 
$C_{y,j} = \{x\in X:\ \alpha_x\ge b^j\cdot\|x - y\|_2^2\}$,
and sort $C_{y,j}$ by nonincreasing order of $\alpha_x - b^j\cdot\|x - y\|_2^2$.
\begin{lemma}\label{lm: testing}
There exists a choice of $Y,y,j$  that satisfies ($i$)--($iv$) iff there exists a choice of 
$y,j$ such that $|C_{y,j}|\ge b^j$ and $C_{y,j}\cap\activ\ne\emptyset$ and the first 
$\min\{|C_{y,j}|,b^{j+1}-1\}$ points in the above order that include $y$ and at least 
one point from $\activ$ are a set that satisfies ($i$)--($iv$).
\end{lemma}

\begin{proof}
Clearly the existence of $y,j$ such that $C_{y,j}$ has the listed properties implies
the existence of $Y,y,j$ that satisfy ($i$)--($iv$). As for the other direction, consider
$Y,y,j$ that satisfy ($i$)--($iv$). Clearly $y\in C_{y,j}$. Suppose that there exists a
point $x\in Y\setminus C_{y,j}$. Then, putting $Y'=Y\setminus \{x\}$ and
$j' = \log_b \floor_b(|Y'|)\le j$, we have that $Y',y,j'$ also satisfy ($i$)--($iv$).
Thus, we may assume that $Y\subset C_{y,j}$. Now, the choice in the lemma
of a subset of $C_{y,j}$ maximizes
$\sum_{x\in Y} \left(\alpha_x - b^j\cdot\|x - y\|_2^2\right)$ subject to the conditions
($i$)--($iii$). Thus, this subset also satisfies ($iv$).
\end{proof}

There are $O(n\log n)$ pairs $y,j$. Listing and sorting each $C_{y,j}$ takes
at most $O(n\log n)$ operations. Listing the candidate $Y\subset C_{y,j}$
and checking it takes $O(|C_{y,j}|)$ operations. Thus, finding a new tight
constraint can be done in polynomial time. (Trivially, we can discretize the
increase of the dual variables and/or use binary search to find the increase
that causes a new constraint to become tight. As we're dealing with squared
Euclidean distance, if the input consists of finite precision rational numbers,
then all computed values are finite precision rational numbers.)

\section{The Algorithm}\label{sec: alg}

We now describe the following three-phase primal-dual algorithm 
(see Algorithm~\ref{fig: alg} on page~\pageref{fig: alg})
that can be used to solve the facility location version of min-sum clustering. In addition
to the pointset $X$, the cluster opening cost $\lambda$, and the target number of 
points $n'$, the algorithm also gets a (sufficiently large, TBD) parameter $b$ that governs
the excess number of discarded outliers in its output. Throughout the algorithm, sets of
points $Y\subset X$ will maintain values $\card_b(Y)$ and $\ctr(Y)$. Clearly, we cannot
do this explicitly and efficiently for every set $Y\subset X$. We use Lemma~\ref{lm: testing}
and its consequences to implement the operations that we need, without storing explicitly
these values for more than $n$ sets. This affects only the first phase of the algorithm. 
For $x\in X$ and $Y\subset X$, we denote throughout the paper 
$d(x,Y) = b^{\card_b(Y)}\cdot\|x - \ctr(Y)\|_2^2$. This is interpreted according to the relevant
values of $\card_b(Y)$ and $\ctr(Y)$.

Phase 1 constructs a dual solution and collects candidate clusters. During phase 1, a point 
$x$ is either active or inactive. Initially, for all $x\in X$, we set $\alpha_x$ to $0$, and we set 
$x$ to be active. The set of candidate clusters $\pclusters$ is empty. We raise all active $x$ 
at a uniform rate, and pause to change the status of points and clusters if one of the following 
events happens.
\begin{itemize}
\item There exists an active $x\in X$ and a cluster $Y\in \pclusters$ such that
         $\alpha_x\ge d(x,Y)$. In this case, replace $Y$ by $Y\cup\{x\}$ in $\pclusters$.
         The new cluster in $\pclusters$ inherits the $\card_b$ and $\ctr$ values from
         $Y$. Also set $x$ to be inactive.
\item There exists $Y\subset X$ that contains an active point and $y\in Y$ such that the dual
         constraint associated with the pair $Y,y$ is tight. Explicitly,
         $$
         \sum_{x\in Y} \alpha_x \ge \lambda + \cost_b(Y),
         $$
         where we set $\card_b(Y) = \log_b\floor_b(|Y|)$ and $\ctr(Y) = y$.
         In this case, add an inclusion-wise minimal such $Y$ to $\pclusters$ and set all 
         $x\in Y$ to be inactive (and set $\card_b(Y)$ and $\ctr(Y)$ as stated above).
\end{itemize}
The first phase ends as soon as the number of active $x\in X$ drops to $n-n'$ or
lower. If this number drops below $n-n'$, we do not add the last cluster $Y_{\last}$
to $\pclusters$, but keep it separately. Note that each new tight constraint causes
at least one point $x\in X$ to become inactive, hence the number of sets $Y$ that
require storing explicitly their parameters $\card_b(Y)$ and $\ctr(Y)$ is at most 
$n'\le n$.

In phase 2, we trim the set of candidate clusters and assign points uniquely to the clusters
in the trimmed list, as follows. Note that we need the parameters $\card_b$ and $\ctr$ only
for clusters for which these values were stored explicitly in phase 1.
Define a conflict graph on the clusters in $\pclusters$. Two
clusters $Y_1,Y_2\in \pclusters$ are connected by an edge in the conflict graph iff there
exists $x\in Y_1\cap Y_2$ such that $\alpha_x > \max\{d(x,Y_1),d(x,Y_2)\}$. In other words,
the edge $\{Y_1,Y_2\}$ indicates that there is $x\in Y_1\cap Y_2$ that contributes to the
opening cost $\lambda$ of both $Y_1$ and $Y_2$. Next, take a lexicographically maximal
independent set ${\cal I}$ in the conflict graph, ordering $\pclusters$ by non-increasing 
order of $\card_b(Y)$, breaking ties arbitrarily. We group the points clustered in $\pclusters$ 
into meta-clusters of the form $(Y,Y')$, where $Y\in{\cal I}$ indicates the meta-cluster, and
$Y'$ is a set of points. (Thus, the entire meta-cluster associated with $Y$ is
$\cup_{(Y,Y')\in\mclusters} Y'$.) In particular, for $Y\in{\cal I}$, we put $(Y,Y)\in\mclusters$.
Any remaining points in $\pclusters$ are added as follows. If $Y'\not\in{\cal I}$, then let
$Y''$ be the set of remaining points in $\{x\in Y':\ \alpha_x = \max_{y\in Y'} \alpha_y\}$,
and let $Y\in{\cal I}$ be such that $Y$ precedes $Y'$ in the order on $\pclusters$ and
$\{Y,Y'\}$ is an edge. Add $(Y,Y'')$ to $\mclusters$, with $\card_b(Y'')=\card_b(Y')$ and
$\ctr(Y'')=\ctr(Y)$. Finally, if fewer than $n'$ points were thus assigned to meta-clusters,
add $(Y_{\last},Y)$ to $\mclusters$, where $Y$ is a set of previously unclustered points
from $Y_{\last}$ of the cardinality needed to complete the number of clustered point to
$n'$. (Notice that at least $n'$ points are clustered in $\pclusters\cup\{Y_{\last}\}$, so this 
is possible.)

Phase 3 determines the final output clustering of the points. For every meta-cluster
$(Y,\cdot)$ and for every integer $p\le\card_b(Y)$, let $n_{Y,p}$ denote the number 
of points $x\in X$ such that there exists $(Y,Y')\in\mclusters$ with $Y'\ni x$ and
$\card_b(Y') = p$. We open clusters as follows. For $p = \card_b(Y)$, we open 
$\max\{1,\left\lfloor \frac{n_{Y,p-2}+n_{Y,p-1}+n_{Y,p}}{b^{2+p}}\right\rfloor\}$ clusters
and assign all the points counted in $n_{Y,p-2},n_{Y,p-1},n_{Y,p}$ to these clusters, 
as evenly as possible. 
\begin{lemma}\label{lm: cardinality max}
The number of points in each such cluster is at most $2b^{2+p}$, and if $Y\ne Y_{\last}$
then this number is at least $b^p$.
\end{lemma}

\begin{proof}
If we open one cluster, then clearly $n_{Y,p-2}+n_{Y,p-1}+n_{Y,p} < 2b^{2+p}$. 
If we open $s > 1$ clusters, then we must have 
$sb^{2+p}\le n_{Y,p-2}+n_{Y,p-1}+n_{Y,p} < (s+1)b^{2+p}$.
Thus, the number of points in each cluster is between $b^{2+p}$ and
$(1+1/s) b^{2+p}$. Clearly, if $Y\ne Y_{\last}$, then $(Y,Y)\in\mclusters$, and
by the definition of $p = \card_b(Y)$, it holds that $|Y|\ge b^p$.
\end{proof}

For $p < \card_b(Y)-2$, we open $\left\lfloor \frac{n_{Y,p}}{b^{2+p}}\right\rfloor$ 
clusters. If this number is at least $1$, we assign all the points counted in $n_{Y,p}$ 
to these clusters, as evenly as possible. If this number is $0$, we discard all
the points counted in $n_{Y,p}$ as outliers.
\begin{lemma}\label{lm: cardinality p}
In this step, if no cluster is opened then the number of points that are
discarded is less than $b^{2+p}$, and otherwise the number of points
in each cluster is at least $b^{2+p}$ and less than $2b^{2+p}$.
\end{lemma}

\begin{proof}
The assertion is trivial.
\end{proof}

We are now ready to define our min-sum $k$-clustering algorithm 
(see Algorithm~\ref{fig: alg2} on page~\pageref{fig: alg2}).
If $k\le \frac{4}{\eps}$, we can run a PTAS or a constant factor approximation for fixed $k$ 
(for instance~\cite{FKKR03,CS07}).\footnote{These papers consider only the case without
outliers. The PTAS in~\cite{FKKR03} enumerates over cluster sizes and approximate cluster 
centers, then computes an optimal assignment of the data points to the approximate centers,
given the corresponding cluster sizes. Clearly, the algorithm can be adapted trivially to handle 
the case with outliers by modifying the target sum of cluster sizes.}
Otherwise, our algorithm follows the general schema 
of the lagrangian relaxation method. Let $\delta > 0$ be determined later. We run the
procedure {\sc PrimalDual} on various values of $\lambda$, and if the smallest returned 
cluster has at most $\frac{\eps}{3}\cdot n'$ points, we remove this cluster. Using binary search 
on the Lagrange multiplier $\lambda$, we find two values $\lambda_1 < \lambda_2$, with 
$\lambda_2 - \lambda_1 < \delta$, that satisfy the following property. The above process
(running {\sc PrimalDual}, then removing the smallest cluster if it's sufficiently small)
returns $k_1 > k$ clusters for $\lambda=\lambda_1$, and $k_2\le k$ clusters for 
$\lambda=\lambda_2$. If $\frac{k-k_2}{k_1-k_2}\ge 1 - \frac{\eps}{4}$, we output the
$k$ largest clusters in the solution for $\lambda_1$, and otherwise we output the solution
for $\lambda_2$.

\begin{theorem}\label{thm: main sqEuc}
The excution of {\bf procedure} {\sc MinSumClustering}$(X,k,n',\eps)$ computes
a clustering of $X'\subset X$ into $k$ clusters such that 
$|X'| \in [(1-\eps) n', n']$, and the total cost of the clustering of $X'$ is at most
$O\left(\frac{1}{\eps^3}\right)\cdot \opt(X,n',k)$.
The time complexity of this computation is $\poly(n,\log(1/\eps),\log\Delta)$,
where $\Delta$ is the ratio of largest to non-zero smallest $\|\cdot\|_2^2$ 
distance in $X$. 
\end{theorem}

\begin{proof}
The performance guarantee is an immediate consequence of 
Corollary~\ref{cor: min-sum clustering}
below. The running time is a straightforward analysis of the code.
\end{proof}

\begin{theorem}\label{thm: main metric}
The same claim applies to instances of points in a metric space $(X,\dist)$,
with $\dist$ replacing $\|\cdot\|_2^2$ in the code and in the claim.
\end{theorem}

\begin{proofsketch}
The $\|\cdot\|_2^2$ distance can be replaced by any metric distance $\dist$
in all claims starting from Proposition~\ref{pr: eps net}. The proofs sometime
require minor changes. In particular, in Lemma~\ref{lm: connection cost}, the 
factor $\frac 1 9$ can be improved to $\frac 1 3$ on account of the triangle
inequality, and this improves all the other constants that depend on it.
\end{proofsketch}

\section{Proofs}\label{sec: proof}

In this section we analyze the min-sum $k$-clustering algorithm. The analysis
builds on the following guarantees of the primal-dual schema.
\begin{theorem}\label{thm: primal-dual}
For every $\eps \in (0,1]$ there exists a constant $c = c_{\eps}$ such that the following
holds. Let $\clusters$ be the output of {\bf procedure} {\sc PrimalDual}($X,\lambda,n',b$),
and let $\alpha$ be the dual solution computed during the execution of this procedure.
Set $\gamma = \max_{x\in X} \alpha_x$. Then,
\begin{enumerate}
\item $(\alpha,\gamma)$ is a feasible solution to the dual program~\eqref{eq: dual-lag-rev}.
\item $\sum_{Y\in\clusters} |Y| \in [(1 - \frac{\eps}{3}) n', n']$.
\item $c\cdot\left(|\clusters|-1\right)\cdot\lambda + 
          \sum_{Y\in\clusters} \cost(Y) \le c\cdot\sum_{Y\in\clusters}\sum_{x\in Y} \alpha_x$.
\end{enumerate}
\end{theorem}

\begin{corollary}\label{cor: min-sum clustering}
Let $\clusters$ be the output of {\bf procedure} {\sc MinSumClustering}. Then, the 
following assertions hold:
\begin{enumerate}
\item $|\clusters|\le k$.
\item $\sum_{Y\in\clusters} |Y|\in [(1-\eps) n',n']$.
\item $\sum_{Y\in\clusters} \cost(Y)\le \frac{8(c+1)}{\eps}\cdot\opt(X,n',k)$.
\end{enumerate}
\end{corollary}

\begin{proof}
The first assertion follows directly from the definition of the procedure. 

For the second assertion,
let $\lambda_i,\clusters_i$ be the values that determine the output of the procedure. By 
Theorem~\ref{thm: primal-dual}, $\sum_{Y\in\clusters_i} |Y| \ge (1 - \frac{\eps}{3}) n'$.
If $Y_{\min,\lambda_i}$ is removed from $\clusters_i$, then 
$|Y_{\min,\lambda_i}|\le\frac{\eps}{3}\cdot n'$. Thus, if $i=2$ then clearly the assertion 
holds. If $i=1$, then $\rho_1\ge 1- \frac{\eps}{4}$. Therefore,
$$
k_1\le \frac{1}{\rho_1}\cdot k \le \left(1 + \frac{\eps}{4-\eps}\right)\cdot k\le 
        \left(1 + \frac{\eps}{3}\right)\cdot k.
$$
Thus, the procedure removes from the output at most a fraction of $\frac{\eps}{3}$ of the 
clusters in $\clusters_1$. As the removed clusters are the smallest, they contain at most 
$\frac{\eps}{3}\cdot n'$ points.

As for the third assertion, consider the two solutions $\clusters_1, \clusters_2$ that are
used to determine the procedure's output. For $i=1,2,$, let $\alpha_i,\pclusters_i$ be the 
output of {\sc PrimalDualPhase1} during the computation of $\clusters_i$. Put
$\gamma_i = \max_{x\in X} \alpha_{i,x}$. Let
$$
(\alpha,\gamma) = \rho_1 (\alpha_1,\gamma_1) +  (1-\rho_1) (\alpha_2,\gamma_2).
$$
Clearly, $(\alpha,\gamma)$ is a feasible solution to the dual LP~\eqref{eq: dual-lag-rev} with the
constant $\lambda = \rho_1\lambda_1+(1-\rho_1)\lambda_2$. Notice that there are exactly 
$n-n'$ points that are not included in $\pclusters_i$. Each point $x\in X$ which is not
included in $\pclusters_i$ has $\alpha_{i,x} = \gamma_i$. (Notice that all the points that
are excluded are active. This is true even for points that are discarded from the last tight
cluster that gets included in $\pclusters_i$.) So, the value of the solution $(\alpha,\gamma)$
is
\begin{eqnarray*}
2\cdot\opt(X,n',k) & \ge & \sum_{x\in X} \alpha_x - (n-n')\gamma - \lambda k \\
& = & \rho_1\cdot\left(\sum_{x\in X} \alpha_{1,x} - (n-n')\gamma_1 - \lambda k_1\right) +
(1-\rho_1)\cdot\left(\sum_{x\in X} \alpha_{2,x} - (n-n')\gamma_2 - \lambda k_2\right) \\
& = & \rho_1\cdot\left(\sum_{Y\in\pclusters_1}\sum_{x\in Y} \alpha_{1,x} - \lambda k_1\right) +
(1-\rho_1)\cdot\left(\sum_{Y\in\pclusters_2}\sum_{x\in Y} \alpha_{2,x} - \lambda k_2\right) \\
& \ge & \rho_1\cdot\left(\sum_{Y\in\pclusters_1}\sum_{x\in Y} \alpha_{1,x} - \lambda_1 k_1\right) +
(1-\rho_1)\cdot\left(\sum_{Y\in\pclusters_2}\sum_{x\in Y} \alpha_{2,x} - \lambda_2 k_2\right)\\ 
 & &  - \delta\cdot(k_1+k_2),
\end{eqnarray*}
where the first inequality follows from Lemma~\ref{lm: duality-rev}, and the first equality uses the 
fact that $k = \rho_1 k_1 + (1-\rho_1) k_2$.

For $i=1,2$ consider the final value of $\clusters_i$. By definition, $k_i$ is one less than
the number of clusters returned from {\bf procedure} {\sc PrimalDual}, so by
Theorem~\ref{thm: primal-dual}, 
$$
\sum_{Y\in\clusters_i} \cost(Y) \le c\cdot\left(\sum_{Y\in\clusters_i}\sum_{x\in Y} \alpha_x - k_i\cdot\lambda_i\right).
$$
In particular, the right-hand side is non-negative. Notice that if $\rho_1 \ge 1 - \frac{\eps}{4}$, then
clearly $\rho_1 > \frac{\eps}{4}$ and the cost of the output clustering is at most
$\sum_{Y\in\clusters_1} \cost(Y)$. Similarly, if $\rho_1 < 1 - \frac{\eps}{4}$, then
$1 - \rho_1 > \frac{\eps}{4}$ and the cost of the output clustering is at most
$\sum_{Y\in\clusters_2} \cost(Y)$. Either way, we get that the cost of the clustering
is at most $\frac{8c}{\eps}\cdot \opt(X,n',k) + \frac{4\delta}{\eps}\cdot(k_1+k_2)\le 
\frac{8(c+1)}{\eps}\cdot \opt(X,n',k)$.
\end{proof}

We now proceed to analyze the primal-dual algorithm and to prove Theorem~\ref{thm: primal-dual}.
The notation follows Algorithm~\ref{fig: alg}.
\begin{lemma}\label{lm: dual values}
At the end of executing {\bf procedure} {\sc PrimalDualPhase1}, for every $Y\in\pclusters$ and 
for every $x\in Y$, we have that $\alpha_x\ge d(x,Y)$.
\end{lemma}

\begin{proof}
When $Y$ is added to $\pclusters$ then there exists $y\in Y$ and $j = \log_b\floor_b(|Y|)$
such that $Y\subset C_{y,j}$. We set $\card_b(Y) = j$ and $\ctr(Y) = y$, so by the definition
of $C_{y,j}$ the lemma holds. If a point $x$ is later added to $Y$, the condition for doing it 
is that $\alpha_x\ge d(x,Y)$.
\end{proof}

\begin{lemma}\label{lm: connection cost}
At the end of executing {\bf procedure} {\sc PrimalDualPhase2}, the following assertions hold:
\begin{enumerate}
\item $\sum_{(Y,Y')\in\mclusters} |Y'| = n'$.
\item For every $(Y,Y')\in \mclusters$ and for every $x\in Y'$, we have that 
        $\alpha_x\ge \frac{1}{9}\cdot d(x,Y')$.
\end{enumerate}
\end{lemma}

\begin{proof}
The first assertion holds as the points that are clustered in $\mclusters$ are all the points that
are clustered in $\pclusters$ plus some of the points clustered in $Y_{\last}$. The number of
such points is at most $n'$ without $Y_{\last}$, and at least $n'$ with $Y_{\last}$. The algorithm
takes from $Y_{\last}$ exactly the number of points needed to complete the number in
$\pclusters$ to $n'$.

For the second assertion, consider $(Y,Y')\in\mclusters$ and $x\in Y'$. If $Y = Y'$, 
then Lemma~\ref{lm: dual values} guarantees the assertion. Otherwise, consider
$Y''\in\pclusters$ that caused $(Y,Y')$ to be added to $\mclusters$. In particuar,
$Y'\subset Y''$, $\card_b(Y') = \card_b(Y'')\le\card_b(Y)$, and there exists 
$z\in Y\cap Y''$ such that $\alpha_z > \max\{d(z,Y),d(z,Y'')\}$. By the choice of
$Y'$ in the algorithm, $\alpha_x = \max_{x'\in Y''} \alpha_{x'}$, so it must be that
$\alpha_z\le\alpha_x$. Notice that by Lemma~\ref{lm: dual values},
$$
\alpha_x\ge d(x,Y'') = b^{\card_b(Y')}\cdot\|x - \ctr(Y'')\|_2^2.
$$
Also 
\begin{eqnarray*}
\alpha_z & > & \max\{b^{\card_b(Y)}\cdot\|z - \ctr(Y)\|_2^2,b^{\card_b(Y'')}\cdot\|z - \ctr(Y'')\|_2^2\} \\
 & \ge & b^{\card_b(Y')}\cdot\max\{\|z - \ctr(Y)\|_2^2,\|z - \ctr(Y'')\|_2^2\}.
\end{eqnarray*}
Thus,
\begin{eqnarray*}
d(x,Y') & = & b^{\card_b(Y')}\cdot\|x - \ctr(Y)\|_2^2 \\
  & \le & b^{\card_b(Y')}\cdot\left(\|x - \ctr(Y'')\|_2+\|z - \ctr(Y'')\|_2+\|z - \ctr(Y)\|_2\right)^2 \\
  & \le & 9\cdot \alpha_x.
\end{eqnarray*}
This completes the proof.
\end{proof}

Let $\clusters$ be the output of {\bf procedure} {\sc PrimalDualPhase3}. Recall that every
cluster $Z\in\clusters$ is derived in some iteration indexed by $(Y,\cdot)\in\mclusters$. It
holds that either $Z\subset Y_{\max}$ or $Z\subset Y_p$ for some $p < \card_b(Y)-2$.
Notice that in the latter case, $Y\ne Y_{\last}$. We will set implicitly $\card_b(Z)$ as follows.
$$
\card_b(Z) = \left\{\begin{array}{ll}
                             \card_b(Y) & Z\subset Y_{\max}, \\
                             p & Z\subset Y_p,\ p < \card_b(Y)-2.
                             \end{array}\right.
$$
We will also set implicitly $\ctr(Z) = \ctr(Y)$.
\begin{lemma}\label{lm: opening cost}
If $Z\subset Y_p$ for some $p < \card_b(Y)-2$, then
$\sum_{x\in Z} \alpha_x \ge b\cdot\lambda$. The same
is true if $Z\subset Y_{\max}$, $Y\ne Y_{\last}$, and
$|Y_{\max}|\ge b^{2+\card_b(Y)}$.
\end{lemma}

\begin{proof}
Consider $x\in Z\subset Y_p$, $p < \card_b(Y)-2$. 
There is a pair $(Y,Y')\in\mclusters$ such that $p = \card_b(Y') < \card_b(Y)-2$,
and $x\in Y'$. Moreover, there is $Y''\in\pclusters$ such that $Y'\subset Y''$ and $\card_b(Y'') = p$
and $\alpha_x = \max_{y\in Y''} \alpha_y$. By the definition of $\card_b$, we have that
$|Y''| < b^{1+p}$. Therefore, $\alpha_x > \frac{\lambda}{b^{1+p}}$. By Lemma~\ref{lm: cardinality p},
$|Z|\ge b^{2+p}$, hence the conclusion. 

A similar argument applies to $Z\subset Y_{\max}$, assuming that $|Y_{\max}|\ge b^{2+\card_b(Y)}$.
In this case, if $Z\supset Y$ then we have $\sum_{x\in Y} \alpha_x\ge \lambda$. As $|Y| < b^{1+\card_b(Y)}$,
$Z$ also contains more than  $b^{2+\card_b(Y)} - b^{1+\card_b(Y)} = b^{2+\card_b(Y)}\cdot\left(1 - \frac 1 b\right)$
points from pairs $(Y,Y')\in\mclusters$, $Y'\ne Y$. By the argument for $Y_p$, for each such point $x$ we have
$\alpha_x > \frac{\lambda}{b^{1+\card_b(Y')}} > \frac{\lambda}{b^{1+\card_b(Y)}}$. Overall, we get that
$\sum_{x\in Z} \alpha_x > \lambda + b^{2+\card_b(Y)}\cdot\left(1 - \frac 1 b\right)\cdot \frac{\lambda}{b^{1+\card_b(Y)}}
= b\cdot\lambda$. If $Z$ does not contain $Y$, then the argument for $Y_p$ holds.
\end{proof}

\begin{lemma}\label{lm: cluster cost}
For every $Z\in\clusters$, $\cost(Z)\le 2b^2\cdot\sum_{x\in Z} d(x,Z)$.
\end{lemma}

\begin{proof}
We have $\cost(Z) = |Z|\cdot\sum_{x\in Z} \|x - \cm(Z)\|_2^2 \le |Z|\cdot\sum_{x\in Z} \|x - \ctr(Y)\|_2^2$.
By Lemmas~\ref{lm: cardinality max} and~\ref{lm: cardinality p},
$|Z|\le 2b^2\cdot\card_b(Z)$.
\end{proof}

\begin{proofof}{Theorem~\ref{thm: primal-dual}}
First, consider the feasibility of $(\alpha,\gamma)$. Clearly, $\gamma$ is set in the theorem to
satisfy the constraints that include it. Regarding the constraints that involve only $\alpha$, we
prove that they are satisfied throughout the execution of {\bf procedure} {\sc PrimalDualPhase1}.
The proof is by induction on the number of inactive points. Clearly, the initial $\alpha$ is feasible.
Now, suppose that $\alpha$ is feasible for some number of inactive points, and consider the
next step when this number increases and a set $A\subset \activ$ is removed from $\activ$
(we will use $\activ$ here to denote the set before the removal of $A$). Let $\alpha'$ denote
the values of the dual variables just before $A$ is removed from $\activ$. If there exist 
$Y\subset X$ and $y\in Y$ such that the constraint for the pair $Y,y$ is violated, then clearly 
$Y\cap \activ\ne\emptyset$, otherwise the same constraint would have been violated by the
solution $\alpha$, as $\alpha$ and $\alpha'$ differ only on $\activ$. But then there is some 
intermediate value $\alpha''$ such that $\alpha''_x = \alpha_x$ for all $x\not\in\activ$ and
$\alpha_x \le \alpha''_x < \alpha'_x$ for all $x\in \activ$, which causes this constraint (or
another constraint involving active points) to become tight. Therefore, at least one point
would have been removed from $\activ$ before we reach the values $\alpha'$, 
in contradiction with our assumptions.

Next, consider the number of points clustered in the output $\clusters$ of {\bf procedure}
{\sc PrimalDual}. Clearly, {\bf procedure} {\sc PrimalDualPhase2} clusters in $\mclusters$
exactly $n'$ points. Some of these points are discarded by {\bf procedure} {\sc PrimalDualPhase3}.
Consider some $(Y,\cdot)\in\mclusters$. By Lemma~\ref{lm: cardinality p}, the number of
points discarded from these clusters is less than
$$
\sum_{p < \card_b(Y)-2} b^{2+p} = \frac{b^{\card_b(Y)} - b^2}{b-1}.
$$
On the other hand, all the points in $Y_{\max}$ are clustered in $\clusters$, as $Y\ne Y_{\last}$
in this case. Clearly, the number of points in $Y_{\max}$ is at least $|Y|\ge b^{\card_b(Y)}$. Thus,
less than $\frac{1}{b-1}\cdot n' \le \eps\cdot n'$ points get discarded.

Finally, let's consider the cost of the clustering. Let $Z\in\clusters$ be a cluster that satisfies
the conditions of Lemma~\ref{lm: opening cost}. Then, 
\begin{eqnarray*}
\sum_{x\in Z} \alpha_x - \lambda & > & \frac{b-1}{b}\cdot \sum_{x\in Z} \alpha_x \\
  & \ge & \frac{b-1}{9b}\cdot \sum_{x\in Z} d(x,Z) \\
  & \ge & \frac{b-1}{18b^3}\cdot\cost(Z),
\end{eqnarray*}
where the first inequality follows from Lemma~\ref{lm: opening cost}, the second inequality
follows from Lemma~\ref{lm: connection cost}, and the third inequality follows from
Lemma~\ref{lm: cluster cost}. The remaining clusters are sets $Y_{\max}$ with
$|Y_{\max}| < b^{2+\card_b(Y)}$ and a subset of $Y_{\last}$. Consider a cluster
$Y_{\max}\in\clusters$. We have that 
$$
\sum_{x\in Y_{\max}} \alpha_x - \lambda = \sum_{x\in Y} \alpha_x - \lambda + \sum_{x\in Y_{\max}\setminus Y} \alpha_x \ge 
\sum_{x\in Y} d(x,Y) + \frac{1}{9}\cdot\sum_{x\in Y_{\max}\setminus Y} d(x,Y)\ge \frac{1}{18b^2}\cdot \cost(Y_{\max}).
$$
Finally, if there's a cluster $Z\subset Y_{\last}$, then 
$$
\sum_{x\in Z} \alpha_x \ge \frac{1}{9}\cdot\sum_{x\in Z} d(x,Y_{\max})\ge \frac{1}{18b^2}\cdot \cost(Z).
$$
Thus, we can set $c = c_{\eps} = \frac{18b^3}{b-1} = \frac{18(1+\eps)^3}{\eps^2}$.
\end{proofof}


\newpage

\begin{algorithm}[H]
\begin{algorithmic}[1]
{\scriptsize
\Procedure{PrimalDual}{$X,\lambda,n',b$} 
\State $\alpha, \pclusters, Y_{\last} \gets${\sc PrimalDualPhase1}($X,\lambda,n',b$)
\State $\mclusters\gets${\sc PrimalDualPhase2}($X,n',b,\alpha,\pclusters,Y_{\last}$)
\State $\clusters\gets${\sc PrimalDualPhase3}($X,b,\mclusters$)
\State \Return $\clusters$
\EndProcedure
\State
%
\Procedure{PrimalDualPhase1}{$X,\lambda,n',b$} 
\State $\activ, \pclusters\gets X, \emptyset$
\State $\alpha_x\gets 0$ for all $x\in X$
\While{$|\activ| > n - n'$}
\State raise $\alpha_x$ at a uniform rate for all $x\in\activ$\Comment{stop raising when one of the following two cases happens}
\If{$\exists x\in\activ$ and $Y\in \pclusters$ such that $\alpha_x\ge d(x,Y)$}
\State $\card_b(Y\cup\{x\}),\ctr(Y\cup\{x\})\gets\card_b(Y),\ctr(Y)$
\State $\pclusters,\activ\gets\pclusters\setminus \{Y\}\cup\{Y\cup\{x\}\},\activ\setminus\{x\}$
\ElsIf{$\exists Y\subset X$ and $y\in Y$ such that $Y\cap\activ\ne\emptyset$ the dual constraint for $Y,y$ is tight}
\State $\card_b(Y), \ctr(Y)\gets\log_b\floor_b(|Y|), y$
\If{$|\activ\setminus Y| < n-n'$} \Comment{use an inclusion-wise minimal such $Y$}
\State return $\alpha, \pclusters, Y$
\Else
\State $\pclusters,\activ\gets \pclusters\cup\{Y\},\activ\setminus Y$
\EndIf
\EndIf
\EndWhile
\State \Return $\alpha, \pclusters, \emptyset$
\EndProcedure	
\State
%
\Procedure{PrimalDualPhase2}{$X,n',b,\alpha,\pclusters,Y_{\last}$}
\State $\activ, \mclusters\gets \{x\in X:\ x\in Y\in \pclusters\vee x\in Y_{\last}\}, \emptyset$
\For{$Y\in \pclusters$, by order of nonicreasing $\card_b(Y)$}
\If{$\exists (Y',Y')\in \mclusters$ with $x\in Y\cap Y'$ and $\alpha_x > \max\left\{d(x,Y),d(x,Y')\right\}$}
\State $Y'',\card_b(Y''),\ctr(Y'')\gets \{x\in Y\cap\activ:\ \alpha_x = \max_{y\in Y} \alpha_y\}, \card_b(Y), \ctr(Y')$
\State $\mclusters\gets\mclusters\cup \{(Y',Y'')\}$
\State $\activ\gets\activ\setminus Y''$
\Else
\State remove each $x\in Y$ from any $Y''\ni x$ with $(Y',Y'')\in\mclusters$
     \Comment{$\alpha_x\le d(x,Y'')$; $\card_b(Y''),\ctr(Y'')$ don't change}
\State $\mclusters\gets\mclusters\cup \{(Y,Y)\}$
\State $\activ\gets\activ\setminus Y$
\EndIf
\EndFor
\State $Y,\card_b(Y),\ctr(Y)\gets \{|\activ| - n + n'$ points in $Y_{\last}\cap\activ\}, \card_b(Y_{\last}), \ctr(Y_{\last})$
\If{$Y\ne\emptyset$}
\State $\mclusters\gets\mclusters\cup\{(Y_{\last},Y)\}$
\EndIf
\State \Return $\mclusters$
\EndProcedure	
\State

%
\Procedure{PrimalDualPhase3}{$X,b,\mclusters$} 
\State $\clusters\gets\emptyset$
\For{$(Y,\cdot)\in\mclusters$}
\State $Y_{\max}\gets\{x\in X:\ \exists Y'\ni x \hbox{ such that } (Y,Y')\in\mclusters\wedge \card_b(Y')\ge\card_b(Y)-2\}$
\State $\clusters\gets\clusters\cup$ 
        {\sc Partition}($Y_{\max},\max\{1,\left\lfloor |Y_{\max}| / b^{2+\card_b(Y)}\right\rfloor\}$)
\For{$p < \card_b(Y)-2$}
\State $Y_p\gets\{x\in X:\ \exists Y'\ni x \hbox{ such that } (Y,Y')\in\mclusters\wedge\card_b(Y')=p\}$
\If{$|Y_p|\ge b^{2+p}$}
\State $\clusters\gets\clusters\cup$ {\sc Partition}($Y_p,\left\lfloor |Y_p| / b^{2+p}\right\rfloor$)
\EndIf
\EndFor
\EndFor
\State \Return $\clusters$
\EndProcedure
\State
\Procedure{Partition}{$S,m$}\Comment{$m\ge 1$}
\State partition $S$ as evenly as possible into $m$ disjoint subsets $S_1,S_2,\dots,S_m$
\State return $\{S_1,S_2,\dots,S_m\}$
\EndProcedure

}
\end{algorithmic} \caption{Algorithm \textsc{Primal-Dual}} \label{fig: alg}
\end{algorithm}

\newpage

\begin{algorithm}[H]
\begin{algorithmic}[1]
\Procedure{MinSumClustering}{$X,k,n',\eps$} 
\State $\lambda_1\gets 0$, $\lambda_2\gets \sum_{x,y\in X} \|x-y\|_2^2$
\State $\clusters_1\gets \{\{x\}:\ x\in X\}$, $\clusters_2\gets X$
\State $b\gets\frac{1+\eps}{\eps}$
\State $\delta\gets \frac{2}{(n+k)\lambda_2}$\Comment{we need $\delta\le \frac{2}{(n+k)\opt(X,n',k)}$}
\While{$\lambda_2-\lambda_1 > \delta$}
\State $\lambda = \frac 1 2\cdot(\lambda_1+\lambda_2)$
\State $\clusters\gets\hbox{PrimalDual}(X,\lambda,n',b)$
\State $Y_{\min,\lambda}\gets$ smallest cluster in $\clusters$
\State $k'\gets |\clusters|-1$
\If{$|Y_{\min,\lambda}|\le\frac{\eps}{3}\cdot n'$}
\State $\clusters\gets\clusters\setminus \{Y_{\min,\lambda}\}$
\EndIf
\If{$k' > k$}
\State $\lambda_1,\clusters_1,k_1\gets\lambda,\clusters,k'$
\Else \Comment{$k'\le k$}
\State $\lambda_2,\clusters_2,k_2\gets\lambda,\clusters,k'$
\EndIf
\EndWhile
\State $\rho_1\gets \frac{k-k_2}{k_1-k_2}$    \Comment{$k_1 > k\ge k_2\ge 0$}
\If{$\rho_1 \ge 1 - \frac{\eps}{4}$}
\State return $\{k$ largest sets in $\clusters_1\}$
\Else
\State \Return $\clusters_2$ \Comment{If $|\clusters_2| < k$, split clusters arbitrarily to get exactly $k$}
\EndIf
\EndProcedure
\end{algorithmic} \caption{Algorithm \textsc{Min-Sum-Clustering}} \label{fig: alg2}
\end{algorithm}

\newpage

\bibliographystyle{alpha}
\bibliography{ref}

\end{document}